\documentclass[11pt,letter]{article}

\usepackage{fullpage}
\usepackage[utf8x]{inputenc}
\usepackage{amsmath, amsthm}
\usepackage{wrapfig,graphicx,amssymb,textcomp,array,amsmath}
\usepackage{yhmath}

\usepackage{color}
\usepackage{changepage}
\usepackage{algpseudocode} 
\algtext*{EndWhile}
\algtext*{EndIf}
\algtext*{EndFor}
\usepackage{algorithm}

\title{Spanning Trees in Multipartite Geometric Graphs
}

\author{Ahmad Biniaz\thanks{School of Computer Science, Carleton University. 
Supported by NSERC.\newline \indent\indent  ahmad.biniaz@gmail.com, \{jit, anil, morin, michiel\}@scs.carleton.ca}\and
Prosenjit Bose\footnotemark[1]\and David Eppstein\thanks{Computer Science Department, University of California, Irvine.
Supported by NSF grant CCF-1228639.\newline\indent\indent eppstein@ics.uci.edu}\and Anil Maheshwari\footnotemark[1]\and Pat Morin\footnotemark[1] \and Michiel Smid\footnotemark[1]}

\date{\today}
\newcommand{\etal}{{\em et al.}}

\newcommand{\canset}[3]{S(#1,#2,#3)}
\newcommand{\excanset}[4]{S(#1,#2,#3,#4)}
\newcommand{\cset}[1]{S_{#1}}
\newcommand{\ncset}[1]{\overline{S_{#1}}}

\newcommand{\vor}[1]{\mathcal{V}(#1)}

\newcommand{\dt}[1]{\mathrm{DT}(#1)}
\newcommand{\fvd}[1]{\mathcal{F}(#1)}
\newcommand{\fvc}[2]{\phi(#1,#2)}
\newcommand{\CH}[1]{\mathrm{CH}(#1)}
\newcommand{\ovl}[1]{\overline{#1}}
\newcommand{\BCP}[2]{\mathrm{bcp}(#1,#2)}
\newcommand{\BFP}[2]{\mathrm{bfp}(#1,#2)}

\newtheorem{lemma}{Lemma}

\newtheorem{theorem}{Theorem}
\newtheorem{observation}{Observation}
\newtheorem*{problem*}{Problem}

\begin{document}

\maketitle

\begin{abstract}
Let $R$ and $B$ be two disjoint sets of points in the plane where the points of $R$ are colored red and the points of $B$ are colored blue, and let $n=|R\cup B|$. A bichromatic spanning tree is a spanning tree in the complete bipartite geometric graph with bipartition $(R,B)$. The minimum (respectively maximum) bichromatic spanning tree problem is the problem of computing a bichromatic spanning tree of minimum (respectively maximum) total edge length. 
\begin{enumerate}
  \item We present a simple algorithm that solves the minimum bichromatic spanning tree problem in $O(n\log^3 n)$ time. This algorithm can easily be extended to solve the maximum bichromatic spanning tree problem within the same time bound. It also can easily be generalized to multicolored point sets.
  \item We present $\Theta(n\log n)$-time algorithms that solve the minimum and the maximum bichromatic spanning tree problems. 
  \item We extend the bichromatic spanning tree algorithms and solve the multicolored version of these problems in $O(n\log n\log k)$ time, where $k$ is the number of different colors (or the size of the multipartition in a complete multipartite geometric graph).
\end{enumerate}
\end{abstract}


\section{Introduction}
\label{introduction-section}
Let $R$ and $B$ be two disjoint sets of points in the plane, and suppose that the points of $R$ are colored red and the points of $B$ are colored blue. A {\em bichromatic spanning tree} on $R\cup B$ is a spanning tree in the {\em complete bipartite geometric graph} $K(R,B)$ with bipartition $(R,B)$. In other words, a bichromatic spanning tree is a spanning tree in which every edge has a red endpoint and a blue endpoint. The {\em minimum bichromatic spanning tree} (MinBST) problem is the problem of computing a spanning tree in $K(R,B)$ whose total edge length is minimum. Similarly, the {\em maximum bichromatic spanning tree} (MaxBST) problem is the problem of computing a spanning tree in $K(R,B)$ whose total edge length is maximum. 
A natural extension of the MinBST and MaxBST problems is to have more than two colors. In this multicolored version, the input points are colored by $k$ colors, and we are looking for a minimum/maximum spanning tree in which the two endpoints of every edge have distinct colors. In other words, we look for a minimum/maximum spanning tree in a complete $k$-partite geometric graph. We call these problems Min-$k$-ST and Max-$k$-ST, respectively. See Figure~\ref{spanning-trees-fig}.

\begin{figure}[htb]
  \centering
\setlength{\tabcolsep}{0in}
  $\begin{tabular}{cccc}
 \multicolumn{1}{m{.25\columnwidth}}{\centering\includegraphics[width=.18\columnwidth]{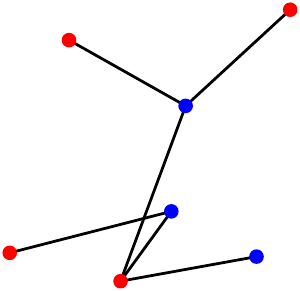}}
&\multicolumn{1}{m{.25\columnwidth}}{\centering\includegraphics[width=.18\columnwidth]{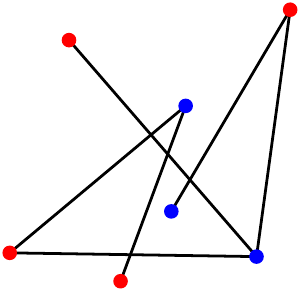}}
& \multicolumn{1}{m{.25\columnwidth}}{\centering\includegraphics[width=.18\columnwidth]{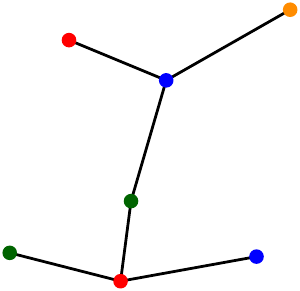}}
&\multicolumn{1}{m{.25\columnwidth}}{\centering\includegraphics[width=.18\columnwidth]{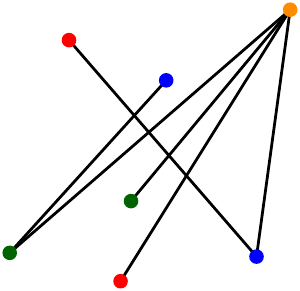}}
\\
MinBST & MaxBST & Min-4-ST & Max-4-ST
\end{tabular}$
  \caption{Colored spanning trees.}
\label{spanning-trees-fig}
\end{figure}
The MinBST and MaxBST problems are natural extensions of the well-known Euclidean minimum and maximum spanning tree problems, which we refer to as the MinST and MaxST problem, respectively, and are defined as follows. Given a set $P$ of $n$ points in the plane, in the MinST (resp. MaxST) problem we seek a spanning tree of minimum (resp. maximum) edge length in the complete geometric graph $K(P)$. 
It is well known that any MinST of $K(P)$ is a subgraph of the Delaunay triangulation of $P$. Thus, one can compute, in $O(n\log n)$ time, a MinST by first computing the Delaunay triangulation and then running Kruskal's or Prim's algorithm on it. Monma and Paterson~\cite{Monma1990} showed that a MaxST of $K(P)$ can also be computed in $O(n\log n)$ time. They also proved a matching lower bound for computing a MaxST
in the algebraic computation-tree model.

For the MinBST and the MaxBST problems, let $n=|R\cup B|$. Then
$K(R,B)$ has $O(n^2)$ edges.
By running Prim's minimum spanning tree algorithm on $K(R,B)$ one can solve the MinBST problem in $O(n^2)$ time. By combining Prim's algorithm with the currently best known dynamic data structure for the bichromatic closest pair problem (see Kaplan~\cite{Kaplan2016}) this problem can be solved in $O(n \log^7 n)$ time. The MaxBST problem can be solved in $O(n^2)$ time by running Prim's algorithm on $K(R, B)$ with negated edge lengths. Because of their geometric nature, we expect faster algorithms for these problems. In this paper, we will study the MinBST, MaxBST, Min-$k$-ST, and Max-$k$-ST problems.
\subsection{Our contribution}
The algorithms presented in this paper are base on Bor\r{u}vka's minimum spanning tree algorithm. In Section~\ref{bi-mst-section} we present a simple algorithm that solves the MinBST problem in $O(n\log^3 n)$ time. We extend this algorithm to solve the MaxBST problem within the same time bound. Also, we extend this algorithm to solve the Min-$k$-ST and Max-$k$-ST problems in $O(n\log^3 n\log k)$ time. By making more use of geometry, in Section~\ref{n-1-section} we present an algorithm that solves the MinBST problem in $O(n\log n)$ time. In Section~\ref{maxbst-section} we show how to adopt the MinBST algorithm and solve the MaxBST problem also in $O(n\log n)$ time. In Section~\ref{k-color-section} we use the MinBST and MaxBST algorithms and solve the Min-$k$-ST and Max-$k$-ST problems in $O(n\log n\log k)$ time.

\section{Bor\r{u}vka's algorithm and binary numbering}
\label{bi-mst-section}
In this section we present a simple algorithm that solves the MinBST and MaxBST problems in $O(n\log^3 n)$ time, where $n$ is the total number of input points. Moreover, we show that this algorithm can be extended to solve the Min-$k$-ST and Max-$k$-ST problems in $O(n\log^3 n\log k)$ time, where $k$ is the number of colors. 
First, we present the algorithm for the MinBST problem, and then we will describe how to extend this algorithm to solve the other problems. Our algorithm is based on Bor\r{u}vka's minimum spanning tree algorithm. Recall that we are given two sets $R$ and $B$ of red and blue points in the plane, respectively, and we want to compute a minimum bichromatic spanning tree (MinBST) in $K(R,B)$. Let $n=|R\cup B|$. 

Bor\r{u}vka's algorithm maintains a forest (initially with each vertex in its own one-node tree) and at each stage of the algorithm adds a set of edges, where each edge is the shortest edge connecting a tree to a vertex outside of it. The number of trees goes down by a factor of two or more in each stage, so there are $O(\log n)$ stages. To use this algorithm for bichromatic point sets, we need to find for every tree $T$ in the current stage, the shortest bichromatic edge that connects a point of $T$ to an oppositely-colored point outside of $T$. We briefly describe how to find, for one stage of Bor\r{u}vka's algorithm, all these shortest bichromatic edges in $O(n\log^2 n)$ time. This implies that Bor\r{u}vka's algorithm solve the MinBST problem in $O(n\log^3 n)$ time. 

To find all shortest bichromatic edges, in the current stage, it is sufficient to solve the following problem: given a partition of the points into subsets (the vertices of the current trees), find for each point $p$ the nearest oppositely-colored point that is not in $p$'s subset. We call this the \emph{all nearest unrelated points} problem. 
We show how to solve the all nearest unrelated points problem in $O(n\log^2 n)$ time. Number the subsets of the partition arbitrarily, as integers from $0$ to one less than the number of subsets, and represent each of these numbers in binary as a sequence of $O(\log n)$ bits. Define a \emph{canonical set} $\canset{i}{b}{c}$ to be the subset of points that belong to a set whose label's $i$th bit is $b$ (for $b\in\{0,1\}$) and whose color is $c$ (for $c\in\{\text{red},\text{blue}\}$). Then for a point $p$, the set of points that are unrelated to $p$ can be expressed as the union of logarithmically many canonical sets $\canset{i}{b}{c}$, one for each possible value of $i$, where $b$ is the complement of the $i$th bit of the label for $p$'s subset, and $c$ is the opposite color to $p$. Therefore, we can solve the all nearest unrelated points problem by computing a Voronoi diagram for each canonical set, and building a point location data structure for each Voronoi diagram. Then, for each point $p$ we query the canonical sets whose union is the set of unrelated points to $p$, and combine the results of the queries to find $p$'s nearest unrelated point. There are $O(\log n)$ Voronoi diagrams, each of which can be built in time $O(n\log n)$, after which we spend $O(\log^2 n)$ time per point to query these diagrams. So the total time for computing all nearest unrelated points is $O(n \log^2 n)$. (It is possible to build the Voronoi diagrams more quickly by computing a single Voronoi diagram of all of the points and using it to guide the construction of the Voronoi diagrams of the subsets---see~\cite{Chazelle2002, Loffler2012}---but this would not speed up the overall algorithm because of the point location time.)

Since Bor\r{u}vka's algorithm takes logarithmically many stages, and each stage can be performed using a single computation of all nearest unrelated points, the total time to construct a bichromatic minimum spanning tree is $O(n \log^3 n)$.

\subsection{Extension to the MaxBST problem}
The combination of Bor\r{u}vka's algorithm and the binary numbering method of the previous section can also be used to solve the MaxBST problem. For this problem, at each stage of Bor\r{u}vka's algorithm we add the longest edges connecting each tree to a vertex outside of it. Therefore, we compute the farthest-point Voronoi diagram for each canonical set, and locate each point $p$ in these diagrams. Since the farthest-point Voronoi diagram of $n$ points can be constructed in $O(n\log n)$ time, the above algorithm solves the MaxBST problem in $O(n \log^3 n)$ time.

\subsection{Extension to the Min-$k$-ST and Max-$k$-ST problems}
\label{extension-section}
In this section we extend the algorithm of the previous section to solve the Min-$k$-ST and Max-$k$-ST problems. We describe the algorithm for the Min-$k$-ST problem; the algorithm for the Max-$k$-ST problem is analogous. 

Recall that in the Min-$k$-ST problem, the input points are colored by $k \geqslant 2$ colors, and we want to compute a minimum spanning tree in which the two endpoints of every edge have distinct colors.  
To use Bor\r{u}vka's algorithm for this problem, we need to find for each point a nearest unrelated point, i.e., a nearest point of different color that is not in its own tree/component. 
Number the trees, and represent each of these numbers in binary as a sequence of $O(\log n)$ bits. Number the $k$ colors by $1,2,\dots, k$, and represent each of these numbers in binary as a sequence of $O(\log k)$ bits. Define a canonical set $\excanset{i}{b}{j}{b'}$ to be the subset of points that (i) their color's $j$th bit is $b'$, and (ii) belong to a tree whose label's $i$th bit is $b$; notice that $b,b'\in\{0,1\}$. Then, for a point $p$, the set of points that are unrelated to $p$ can be expressed as the union of canonical sets $\excanset{i}{b}{j}{b'}$, one for each possible pair $(i,j)$, where $b$ is the complement of the $i$th bit of the label for $p$'s tree, and $b'$ is the complement of the $j$th bit of $p$'s color. To find a nearest unrelated point to $p$, we locate $p$ in the Voronoi diagrams of the canonical sets whose union is the set of unrelated points to $p$. There are $O(\log n\log k)$ canonical sets. It takes $O(n\log^2 n\log k)$ time to build the Voronoi diagrams for all these sets, and also to locate all the points in these diagrams. Thus, the total time to solve the Min-$k$-ST problem (and also the Max-$k$-ST problem) is $O(n \log^3 n\log k)$.

\section{The minimum bichromatic spanning tree problem}
\label{n-1-section}
Recall that we are given two sets $R$ and $B$ of red and blue points in the plane, respectively, and we want to compute a minimum spanning tree in in $K(R,B)$. In this section, we present an algorithm that computes a MinBST in $K(R,B)$ in $O(n\log n)$ time, where $n=|R\cup B|$. In fact, we show how to find, for all stages of Bor\r{u}vka's algorithm, all shortest bichromatic edges in $O(n\log n)$ time. Our algorithm is optimal because finding the bichromatic closest pair has an $\Omega(n\log n)$ lower bound (see \cite{Avis1980}). Our algorithm is summarized below. 

\begin{algorithm}[H]                 
\caption{MinBST$(R,B)$}          
\label{alg1} 
\begin{algorithmic}[1]
    \State Connect each point to a closest point of opposite color.
    \State Run Bor\r{u}vka's algorithm on the resulting components of step 1.
\end{algorithmic}
\end{algorithm}
\begin{observation}
\label{closest-point-obs}
In any minimum bichromatic spanning tree, every point is connected to a closest point of opposite color.  
\end{observation}

The correctness of algorithm MinBST follows from Observation~\ref{closest-point-obs} and from the correctness of Bor\r{u}vka's algorithm. Before analyzing the running time of this algorithm we introduce some notation.
For a point set $Q$ in the plane, let $\dt{Q}$ denote the Delaunay triangulation of $Q$, and $\vor{Q}$ denote the Voronoi diagram of $Q$. Let $\nu(q,Q)$ denote the Voronoi region/cell of a point $q\in Q$ in $\vor{Q}$. For two disjoint point sets $Q_1$ and $Q_2$, where each of the points in $Q_1\cup Q_2$ is colored either red or blue, we define the bichromatic closest pair $\BCP{Q_1}{Q_2}$ as a closest red-blue pair between $Q_1$ and $Q_2$.

Step 1 in algorithm MinBST takes $O(n \log n)$ time; a straightforward solution is to locate all points of $R$ in $\vor{B}$ and all points of $B$ in $\vor{R}$. In the rest of this section we show how to run Bor\r{u}vka's algorithm (step 2) in $O(n\log n)$ time. Notice that, in algorithm MinBST, one can combine steps 1 and 2 and just run Bor\r{u}vka's algorithm on the input point set. However, having step 1 separately, will simplify the running time analysis.   

Recall that Bor\r{u}vka's algorithm maintains a forest (a set of trees which we call components) and at each stage adds the shortest edges connecting each component to a vertex outside of it. Thus, in each stage the number of components goes down by a factor of two or more. The output of the last stage is a single component which is a minimum spanning tree. Therefore, there are $O(\log n)$ stages in total. 
Consider one stage of Bor\r{u}vka's algorithm. Let $C_1,\dots, C_k$ be the input components to that stage. Let $P=R\cup B$, and for each $i\in\{1,\dots,k\}$, let $P_i$ denote the set of points in $C_i$. Note that, after step 1, we have $k\leqslant n/2$. Moreover, each point is in the same component as its closest point of opposite color. We have to find for each component $C_i$, the shortest edge connecting a point in $C_i$ to an oppositely-colored point outside of $C_i$.
In fact, we have to solve the following problem which we call ``all bichromatic closest pairs'': 

\begin{problem*}
Given a set $P$ of red and blue points in the plane and a partition of $P$ into $\{P_1,\dots,P_k\}$ such that each point of $P$ is in the same set $P_i$ as its closest point of opposite color. Find for each $i\in\{1,\dots, k\}$, the closest red-blue pair between $P_i$ and $P\setminus P_i$, i.e., $\BCP{P_i}{P\setminus P_i}$. 
\end{problem*}

Let $R_i$ and $B_i$ denote the set of red and blue points of $P_i$, respectively. Then, $\BCP{P_i}{P\setminus P_i}$ can be computed by taking the shorter of $\BCP{B_i}{R\setminus R_i}$ and $\BCP{R_i}{B\setminus B_i}$.
The following algorithm finds $\BCP{B_i}{R\setminus R_i}$ for all $i\in\{1,\dots,k\}$. By swapping the role of red and blue points, one can compute $\BCP{R_i}{B\setminus B_i}$ for all $i$.

\begin{algorithm}[H]                 
\caption{All-Blue-BCP$(\{R_1\cup B_1, \dots,R_k\cup B_k\})$}          
\label{alg2} 
\begin{algorithmic}[1]
    \State Construct $\dt{R}$.
    \For {$i= 1$ to $k$}
	\State compute $T_i=\{p\in R\setminus R_i\colon$ in $\dt{R}$, $p$ is adjacent to a point in $R_i\}$;
	\State construct $\dt{B_i\cup T_i}$;
	\State $\BCP{B_i}{R\setminus R_i}=$ the endpoints of a shortest red-blue edge in $\dt{B_i\cup T_i}$. 
    \EndFor
\end{algorithmic}
\end{algorithm}

We prove the correctness of algorithm All-Blue-BCP by a non-trivial extension of the proof of Theorem 3.1 in~\cite{AAggarwal1992} to our bichromatic setting.
The set $T_i$, that is computed in line 3 in algorithm All-Blue-BCP, contains the points of $R\setminus R_i$ that are adjacent to some point of $R_i$ in $\dt{R}$; see Figure~\ref{Voronoi-cell-fig}(a). To simplify the notation, we write $\ovl{R_i}$ for $R\setminus R_i$. In line 5, the algorithm computes $\BCP{B_i}{\ovl{R_i}}$ as the endpoints of a shortest red-blue edge in $\dt{B_i\cup T_i}$. Thus, we have to prove that $\BCP{B_i}{\ovl{R_i}}=\BCP{B_i}{T_i}$. Take any $i\in\{1,\dots,k\}$. Let $(b,p)=\BCP{B_i}{\ovl{R_i}}$ where $b\in B_i$ and $p\in \ovl{R_i}$. To prove the correctness of this algorithm, we just need to show that $p$ is in $T_i$. In order to show this, we prove that for some point $q\in R_i$, $pq$ is an edge of $\dt{R}$; this guarantees that algorithm All-Blue-BCP adds $p$ to $T_i$ in line 3. 
\begin{figure}[htb]
  \centering
\setlength{\tabcolsep}{0in}
  $\begin{tabular}{cc}
 \multicolumn{1}{m{.55\columnwidth}}{\centering\includegraphics[width=.45\columnwidth]{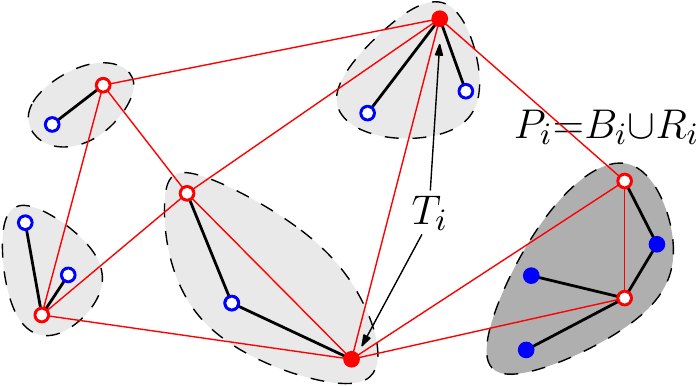}}
&\multicolumn{1}{m{.45\columnwidth}}{\centering\includegraphics[width=.4\columnwidth]{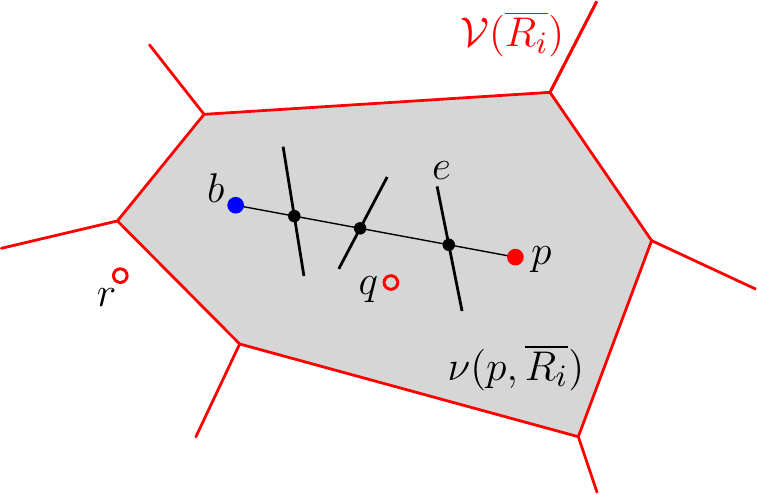}}
\\
(a) & (b) 
\end{tabular}$
  \caption{Proof of the correctness of algorithm All-Blue-BCP: (a) The input components to one stage of Bor\r{u}vka's algorithm. Solid blue points belong to $B_i$ and solid red points belong to $T_i$. (b) $p$ is a point in $\ovl{R_i}$ that is closest to $b$, and $r$ is a point in $R$ that is closest to $b$.}
\label{Voronoi-cell-fig}
\end{figure}

Consider the Voronoi diagram $\vor{\ovl{R_i}}$. Because $p$ is a point of $\ovl{R_i}$ that is closest to $b$, $b$ lies in the Voronoi cell $\nu(p, \ovl{R_i})$. Since $\nu(p, \ovl{R_i})$ is convex, the line segment $bp$ is inside this cell; see Figure~\ref{Voronoi-cell-fig}(b). Thus, no edge of $\vor{\ovl{R_i}}$ crosses $bp$. 
Recall that $b$'s closest red point in $R$, say $r$, is in the same component as $b$. Thus, $r\in R_i$ and also $r\neq p$. Now, imagine the construction of $\vor{R}$ by inserting the points of $R_i$ to $\vor{\ovl{R_i}}$. Since $r$ is the closest red point to $b$, $b$ lies in $\nu(r, R)$. Because of this, and since $\nu(r, R)$ and $\nu(p,R)$ are two different cells in $\vor{R}$, there are some edges of $\vor{R}$ that cross $bp$; see Figure~\ref{Voronoi-cell-fig}(b). Among those edges, consider the edge $e$ whose intersection with $bp$ is closest to $p$. Let $q\in R$ be the point such that $e$ is a common edge between $\nu(q, R)$ and $\nu(p,R)$; notice that $q\neq p$. By inserting the points of $R_i$ to $\vor{\ovl{R_i}}$, the Voronoi cells of $\vor{\ovl{R_i}}$ do not get larger. This implies that the point $q$\textemdash which has $e$ on its boundary in $\nu(q,R)$\textemdash belongs to $R_i$. Therefore, in $\dt{R}$, there is an edge between $p$ and $q$; that edge corresponds to $e$ in $\vor{R}$. Notice however that $q$ is not necessarily $r$ itself. Moreover, if $b$ is on the boundary of $\nu(p, R)$, then $e$ passes through $p$, and $q=r$. This finishes the proof for the correctness of algorithm All-Blue-BCP.

\begin{lemma}
\label{size-T-lemma}
 For any $k$, with $2\leqslant k\leqslant |R|$, and any partition $\{R_1, \dots,R_k\}$ of $R$, the total size of the sets $T_1,\dots, T_k$, that are computed in line 3 of algorithm All-Blue-BCP, is $O(|R|)$. Moreover, having $\dt{R}$, the sets $T_1\dots, T_k$ can be computed in $O(|R|)$ time.
\end{lemma}
\begin{proof}
For each $p\in R$ and each $i\in\{1,\dots,k\}$ we define
\[ f(p,i) =
  \begin{cases}
    1       & \quad \text{if } p\in T_i\\
    0       & \quad \text{if } p\notin T_i.\\
  \end{cases}
\]

The number of sets $T_i$ that a red point $p$ belongs to is at most $p$'s degree $\deg(p)$ in $\dt{R}$. Thus, 
\[\sum_{i=1}^{k}|T_i|=\sum_{i=1}^{k}\sum_{p\in R}f(p,i)=\sum_{p\in R}\sum_{i=1}^{k}f(p,i)\leqslant \sum_{p\in R}\deg(p)\leqslant 2(3|R|-6),
\]
where the last inequality is valid because $\dt{R}$ is a planar graph and has at most $3|R|-6$ edges. Therefore, the total size of the sets $T_1,\dots,T_k$ is $O(|R|)$. Moreover, these sets can be computed in $O(|R|)$ time as follows. Take any edge $pq$ of $\dt{R}$. If $p$ and $q$ belong to a same set $R_i$, then do nothing. If $p$ and $q$ belong to two different sets of the partition, say $p\in R_i$ and $q\in R_j$, then add $p$ to $T_j$ and add $q$ to $T_i$. 
\end{proof}

Now we analyze the running time of algorithm All-Blue-BCP. It takes $O(|R|\log |R|)$ time to compute $\dt{R}$. Since the total size of the sets $B_1,\dots,B_k$ is $|B|$, and, by Lemma~\ref{size-T-lemma}, the total size of the sets $T_1,\dots,T_k$ is $O(|R|)$, we can compute $\dt{B_i\cup T_i}$ for all $i\in\{1,\dots,k\}$ in $O(n\log n)$ time; recall that $n=|R\cup B|$. Thus the total running time of algorithm All-Blue-BCP is $O(n\log n)$, and hence, the running time of algorithm MinBST is $O(n\log^2 n)$. In the rest of this section, we show how to improve the running time of MinBST to $O(n\log n)$. 

Kirkpatrick~\cite{Kirkpatrick1979} shows how to {\em merge} two Delaunay triangulations in linear time; that is, given $\dt{P}$ and $\dt{Q}$, how to find $\dt{P\cup Q}$ in time $O(|P|+|Q|)$. L{\"{o}}ffler and Mulzer~\cite{Loffler2012} show that the reverse operation can also be done in linear time. That is, one can {\em split} $\dt{P\cup Q}$ to obtain $\dt{P}$ and $\dt{Q}$ in time $O(|P|+|Q|)$. We use these two results, and show how to run algorithm All-Blue-BCP during all stages of Bor\r{u}vka's algorithm in total $O(n\log n)$ time. In order to show this, we use the following fact: if we have $\dt{R}$, $\dt{B_i}$ and $\dt{T_i}$ for all $i\in\{1,\dots,k\}$, then one execution of algorithm All-Blue-BCP (during one stage of Bor\r{u}vka's algorithm) takes $O(n)$ time; this is because the total size of the sets $B_1,\dots,B_k$ is $|B|$, and by Lemma~\ref{size-T-lemma} the sets $T_1,\dots,T_k$ can be computed in $O(|R|)$ time and their total size is $O(|R|)$, and thus, by the result of \cite{Kirkpatrick1979} we can compute $\dt{B_i\cup T_i}$ for all $i\in\{1,\dots,k\}$ and also all bichromatic closest pairs in $O(|R|+|B|)$ time. 

We compute $\dt{R}$ at the beginning of algorithm MinBST. As discussed earlier we can compute $\dt{B_i}$ and $\dt{T_i}$ for all input components of the first stage of Bor\r{u}vka's algorithm in $O(n\log n)$ time. Based on the discussion above, we are going to show how to retrieve all $\dt{B_i}$ and $\dt{T_i}$ for input components of the next stage of Bor\r{u}vka's algorithm from all $\dt{B_i}$ and $\dt{T_i}$ of input components of the current stage of Bor\r{u}vka's algorithm. Although for simplicity we use index $i$ to refer to the components of both the current and the next stages, notice that the number of components and their sizes vary from one stage to another. 

Consider one stage of Bor\r{u}vka's algorithm, and let $k$ be the number of input components to this stage, where $2\leqslant k\leqslant n/2$. Without loss of generality let $\{C_1,\dots, C_t\}$, with $2\leqslant t\leqslant k$, denote a subset of the input components that should be connected together and be passed to the next stage as a single component. Let $C^*$ denote the resulting component. Let $R^*$ and $B^*$ denote the sets of red and blue points of $C^*$, respectively. Let $T^*$ be the point set that will be computed (with respect to $R^*$) in line 3 of All-Blue-BCP in the next stage. Notice that $B^*$ is the union of the sets $B_1, \dots, B_t$, and $R^*$ is the union of the sets $R_1, \dots, R_t$. Recall that $T^*$ contains the points of $R\setminus R^*$ that have a Delaunay neighbor in $R^*$.
Thus, $T^*$ is the union of the sets $T_1, \dots, T_t$ minus the set $R^*$.

In order to compute $\dt{B^*}$ for the next stage, we iteratively merge the Delaunay triangulations of the two smallest sets among $B_1, \dots, B_t$. A monotone priority queue (see~\cite{Cherkassky1999}) can be used to find the two smallest sets iteratively; the total time for the queue operations is proportional to the sum of the number of sets and the size of the largest set.
We compute $\dt{T^*}$ during the same merge process that we compute $\dt{B^*}$, i.e., whenever we merge $\dt{B_i}$ and $\dt{B_j}$, we also merge $\dt{T_i}$ and $\dt{T_j}$. Let $R_{ij}=R_i\cup R_j$ and $B_{ij}=B_i\cup B_j$. Let $T_{ij}$ be the set that will be computed in line 3 of All-Blue-BCP with respect to $R_{ij}$. By the result of \cite{Kirkpatrick1979} we can compute $\dt{B_{ij}}$ in $O(|B_i|+|B_j|)$ time by merging $\dt{B_i}$ and $\dt{B_j}$. We describe in more detail how to compute $T_{ij}$ and also $\dt{T_{ij}}$. Let $\ovl{T_{ij}}=(T_i\cup T_j)\cap R_{ij}$. We compute $\ovl{T_{ij}}$ as follows: take any point $p\in R_i$, if, in $\dt{R}$, $p$ is adjacent to a point $q$ in $R_j$, then add both $p$ and $q$ to $\ovl{T_{ij}}$. Then, $T_{ij}=(T_i\cup T_j)\setminus \ovl{T_{ij}}$.
By the result of \cite{Kirkpatrick1979} we can compute $\dt{T_i\cup T_j}$ in $O(|T_i|+|T_j|)$ time by merging $\dt{T_i}$ and $\dt{T_j}$. Then, by the result of \cite{Loffler2012} we can compute $\dt{T_{ij}}$ (and also $\dt{\ovl{T_{ij}}}$) in $O(|T_i\cup T_j|)$ time by splitting $\dt{T_i\cup T_j}$. 

We analyze the total running time of step 2 of algorithm MinBST as follows. Let $B_1,\dots, B_k$ be the sets of blue points of components obtained in step 1 of algorithm MinBST.
Imagine a binary tree $\mathcal{T}$ that is obtained as follows. $\mathcal{T}$ has $k$ leaves that are labeled $B_1,\dots, B_k$. Recall that in algorithm All-Blue-BCP, we merge the Delaunay triangulations in a binary manner. Each time we merge two Delaunay triangulations $\dt{B_i}$ and $\dt{B_j}$, we create a node with label $B_i\cup B_j$ in $\mathcal{T}$ and connect the nodes $B_i$ and $B_j$ as its children. The label of the root of $\mathcal{T}$ is $B$. 

\begin{lemma}
\label{height-lemma}
The height of $\mathcal{T}$ is $O(\log n)$.
\end{lemma}
\begin{proof}
Take any leaf $B_i$ of $\mathcal{T}$. Let $\mathcal{P}$ be the path from $\mathrm{root}(\mathcal{T})$ to 
\let\qed\relax\end{proof}
\begin{wrapfigure}{r}{1.8in} 
\vspace{-25pt} 
\includegraphics[width=1.8in]{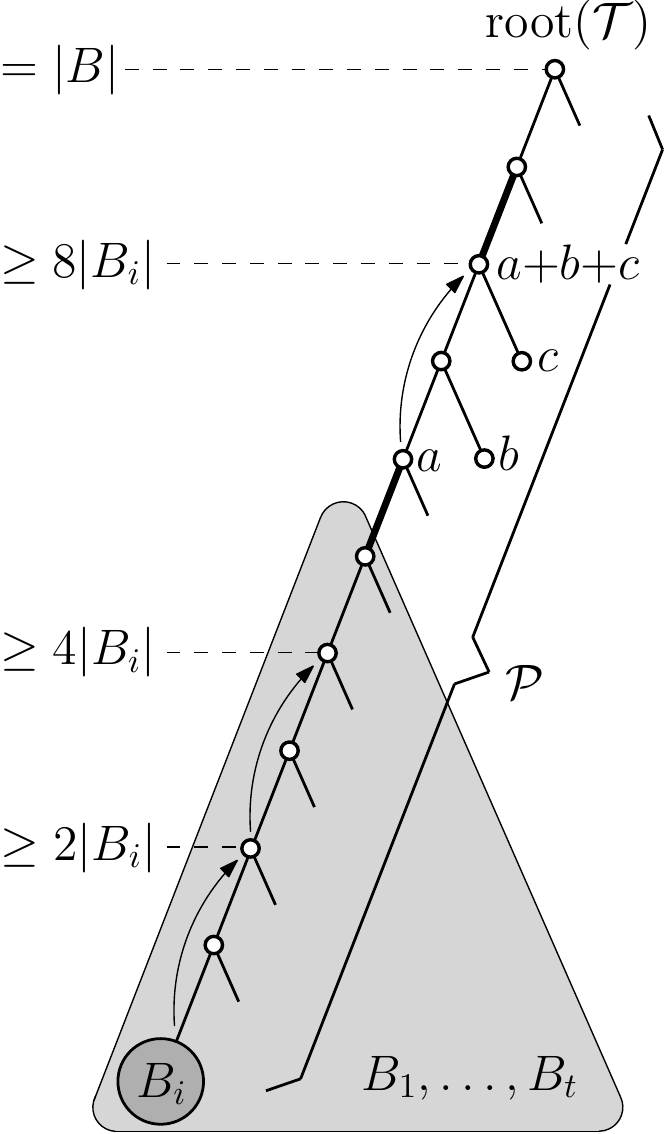} 
\vspace{-15pt} 
\end{wrapfigure}
\vspace{-8pt}
\noindent
 $B_i$. Partition the edges of $\mathcal{P}$ into $\{\mathcal{P}_1,\mathcal{P}_2\}$, where $\mathcal{P}_1$ is the set of edges connecting a node in one stage of Bor\r{u}vka's algorithm to a node in the next stage, and $\mathcal{P}_2$ is the set of edges connecting two nodes within one stage of Bor\r{u}vka's algorithm. See the figure to the right; the fat edges belong to $\mathcal{P}_2$, while the other edges belong to $\mathcal{P}_1$. Since the number of stages in Bor\r{u}vka's algorithm is $\log n$, the number of edges in $\mathcal{P}_1$ is $\log n$. Imagine walking up $\mathcal{P}$, starting from $B_i$ and ending at $\mathrm{root}(\mathcal{T})$. Each time we go up $\mathcal{P}$ and pass two consecutive edges of $\mathcal{P}_2$, the size of the label of that node gets doubled. To see why this is true, consider two consecutive merges of a node $a$ with nodes $b$ and $c$, respectively; see the figure to the right. Here, we use the label of a node to refer to the size of that node as well. Recall that within each stage we iteratively merge the two smallest components, thus, the reason that $a$ was merged with $b$ before $c$ is that none of $a$ and $b$ are bigger than $c$, that is $a\leqslant c$ and $b\leqslant c$. Thus, the size of the grandparent of $a$, which is $a+b+c$, is at least $2a$. Since the label of $\mathrm{root}(\mathcal{T})$ has size at most $n$, the doubling process won't repeat more than $\log n$ times. Thus, the number of edges in $\mathcal{P}_2$ is at most $2\log n+\log n$ (the extra $\log n$ is for the cases when the parity of the number of edges within some stages are odd). Therefore $\mathcal{P}$ has at most $4\log n$ edges. This implies that the height of $\mathcal{T}$ is $O(\log n)$. 
\qed
\vspace{10pt}

By Lemma~\ref{height-lemma}, $\mathcal{T}$ has $O(\log n)$ levels. The labels of the nodes of each level is a partition of a subset of $B$. Accordingly, we have a partition of a subset of $R$ with respect to those nodes. Thus, the total number of blue points (resp. red points) in each level is at most $|B|$ (resp. at most $|R|$). Therefore, the time we spend per level of $\mathcal{T}$ to compute all $\dt{B_i}$ from the Delaunay triangulations of their children is $O(|B|)$. Lemma~\ref{size-T-lemma} can easily be generalized to any partition of any subset of $R$. Based on this, the total size of sets $T_i$ in each level of $\mathcal{T}$ is $O(|R|)$. Therefore, the time we spend per level to compute all $\dt{T_i}$ from the Delaunay triangulations of their children is $O(|R|)$. It turns out that for each level of $\mathcal{T}$ we spend $O(n)$ time. Therefore, the total running time of algorithm MinBST is $O(n\log n)$.

\begin{theorem}
 \label{MinBST-thr}
Given two disjoint sets $R$ and $B$ of points in the plane, a Euclidean minimum spanning tree in $K(R,B)$ can be computed in $\Theta(n\log n)$ time, where $n=|R\cup B|$. 
\end{theorem}

\section{The maximum bichromatic spanning tree problem}
\label{maxbst-section}
In this section we consider the MaxBST problem. Given $R$ and $B$, we present an algorithm that computes a MaxBST, in $K(R,B)$, in $O(n\log n)$ time, where $n=|R\cup B|$. Our algorithm is optimal because fining the bichromatic farthest pair has an $\Omega(n\log n)$ lower bound (see \cite[Theorem 4.16]{Preparata1985}). Our MaxBST algorithm is the same as the MinBST algorithm that was presented in Section~\ref{n-1-section}, except, in step 1 we connect each point to a farthest point of opposite color, and in step 2, at each stage of Bor\r{u}vka's algorithm, we add the longest red-blue edges connecting each component to a vertex outside of it. The correctness of this algorithm follows from the correctness of Bor\r{u}vka's algorithm and the fact that in any maximum bichromatic spanning tree, every point is connected to a farthest point of opposite color. We skip the parts of MaxBST that are similar to MinBST, but describe its core parts that lead to the same running time.

Before analyzing the running time we introduce some notation, and present a lemma that plays an important role in the analysis.
For a point set $Q$ in the plane let $\CH{Q}$ denote the list of the vertices of the convex hull of $Q$, ordered along the boundary, and let $\fvd{Q}$ denote the farthest point Voronoi diagram of $Q$. Let $\fvc{q}{Q}$ denote the farthest point Voronoi region/cell of a point $q\in Q$ in $\fvd{Q}$; notice that $q$ has a non-empty Voronoi cell in $\fvd{Q}$ if and only if $q$ is in $\CH{Q}$. For two disjoint point sets $Q_1$ and $Q_2$, where each of the points in $Q_1\cup Q_2$ is colored either red or blue, we define the bichromatic farthest pair $\BFP{Q_1}{Q_2}$ as a farthest red-blue pair between $Q_1$ and $Q_2$.

\begin{lemma}
 \label{MaxBST-lemma}
Let $R$ and $B$ be two disjoint sets of points in the plane, and let $T$ be a maximum spanning tree of $K(R,B)$. Then for every edge $(r,b)\in T$ with $r\in R$ and $b\in B$, $r$ is in the convex hull of $R$ or $b$ is in the convex hull of $B$.
\end{lemma}

\begin{proof}
The proof is by contradiction; take an edge $(r,b)\in T$ with  $r\in R$ and $b\in B$, where $r\notin\CH{R}$ and $b\notin\CH{B}$. Without loss of generality assume $(r,b)$ is horizontal, and $r$ is to the left of $b$. 
\let\qed\relax\end{proof}
\vspace{-10pt}
\begin{wrapfigure}{r}{1.6in} 
\vspace{-8pt} 
\includegraphics[width=1.5in]{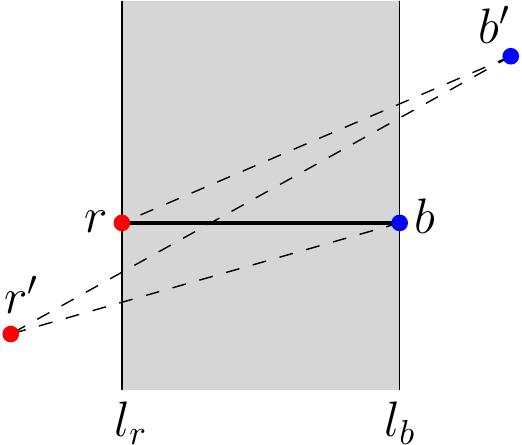} 
\vspace{-5pt} 
\end{wrapfigure}
Consider two vertical lines $l_r$ and $l_b$ that pass through $r$ and $b$, respectively. Since $r\notin \CH{R}$, there is a point $r'$ that belong to $R$ and is to the left of $l_r$. Similarly, there is a point $b'$ that belong to $B$ and is to the right of $l_b$. Observe that $|rb|$ is smaller than each of $|rb'|$, $|r'b|$, and $|r'b'|$. Let $T_r$ and $T_b$ be the two trees obtained from $T$ by removing $(r,b)$. We consider two cases: (i) $r'$ and $b'$ are in different trees, or (ii) $r'$ and $b'$ are in a same tree. In case (i) we obtain a tree $T'$ by joining $T_r$ and $T_b$ with $(r',b')$. In case (ii) without loss of generality assume that $r'$ and $b'$ are in $T_r$. Then, we obtain a tree $T'$ by joining $T_r$ and $T_b$ with $(r',b)$. In both cases, $T'$ is a spanning tree of $K(R,B)$ that is longer than $T$; this contradicts the optimality of $T$.
\qed\newline

\begin{observation}
 \label{bfp-obs}
Let $R$ and $B$ be two sets of points in the plane. Then, the longest edge between $R$ and $B$ has an endpoint in $\CH{R}$ and an endpoint in $\CH{B}$. 
\end{observation}

Step 1 in algorithm MaxBST takes $O(n \log n)$ time; a straightforward solution is to locate all points of $R$ in $\fvd{B}$ and all points of $B$ in $\fvd{R}$. In the rest of this section we show how to run Bor\r{u}vka's algorithm (step 2) in $O(n\log n)$ time.   
Let $C_1,\dots, C_k$ be the components of the current stage of Bor\r{u}vka's algorithm. We have to find for each component $C_i$, the longest edge connecting a point in $C_i$ to an oppositely-colored point outside of $C_i$. In fact, we have to solve the following problem: 

\begin{problem*}
Let $R$ and $B$ be two sets of red and blue points in the plane, respectively. Given a partition of $R\cup B$ into $\{R_1\cup B_1,\dots,R_k\cup B_k\}$ such that each point is in the same set $R_i\cup B_i$ as its farthest point of opposite color, find for each $i\in\{1,\dots, k\}$, the farthest red-blue pair between $B_i$ and $R\setminus R_i$, and the farthest red-blue pair between $R_i$ and $B\setminus B_i$, i.e., $\BFP{B_i}{R\setminus R_i}$ and $\BFP{R_i}{B\setminus B_i}$. 
\end{problem*}

The following algorithm finds $\BFP{B_i}{R\setminus R_i}$ for all $i\in\{1,\dots,k\}$. By swapping the role of red and blue points, one can compute $\BFP{R_i}{B\setminus B_i}$ for all $i$. For the purpose of this algorithm, by Lemma~\ref{MaxBST-lemma}, we assume that $R$ is in convex position; this can be done by setting $R$ to be $\CH{R}$. This assumption does not violate the fact that each blue point is in the same component as its farthest red point, because the farthest red point to any blue point is in $\CH{R}$. By Observation~\ref{bfp-obs}, the longest edge between $B_i$ and $R\setminus R_i$ has an endpoint in $\CH{B_i}$. Thus, to compute $\BFP{B_i}{R\setminus R_i}$, we maintain $\CH{B_i}$ and compute $\BFP{\CH{B_i}}{R\setminus R_i}$ instead.

\begin{algorithm}[H]                 
\caption{All-Blue-BFP$(\{R_1\cup B_1, \dots,R_k\cup B_k\})$}          
\label{alg3} 
\begin{algorithmic}[1]
    \State Construct $\fvd{R}$.
    \For {$i= 1$ to $k$}
	\State $T_i=\{p\in R\setminus R_i\colon$ in $\fvd{R}$, the cell of $p$ is adjacent to the cell of a point of $R_i\}$;
	\State construct $\CH{B_i}$ and $\CH{T_i}$;
	\State $\BFP{B_i}{R\setminus R_i}=$ the points defining the maximum distance between $\CH{B_i}$ and $\CH{T_i}$. 
    \EndFor
\end{algorithmic}
\end{algorithm}

First, we prove the correctness of algorithm All-Blue-BFP. To simplify the notation, we write $\ovl{R_i}$ for $R\setminus R_i$.
Since $R$ is in convex position, every point of $R$ has a non-empty Voronoi cell in $\fvd{R}$. The set $T_i$, that is computed in line 3, contains the points of $\ovl{R_i}$ whose cell in $\fvd{R}$ is adjacent to the cell of a point in $R_i$. Then, in line 5, the
algorithm computes $\BFP{B_i}{\ovl{R_i}}$ as the endpoints of the longest red-blue edge between $\CH{B_i}$ and $\CH{T_i}$,
i.e., $\BFP{B_i}{\ovl{R_i}}=\BFP{\CH{B_i}}{\CH{T_i}}$. Notice that $T_i$ is in convex position. Take any $i\in\{1,\dots,k\}$. Let $(b,p)=\BFP{B_i}{\ovl{R_i}}$ where $b\in B_i$ and $p\in \ovl{R_i}$. To prove the correctness of this algorithm, we only need to show that $p$ is in $T_i$. In order to show this, we prove that $\fvc{p}{R}$ is adjacent to $\fvc{q}{R}$ for some point $q\in R_i$; this guarantees that algorithm All-Blue-BFP adds $p$ to $T_i$ in line 3. 

\begin{figure}[htb] 
 \centering
\includegraphics[width=2.6in]{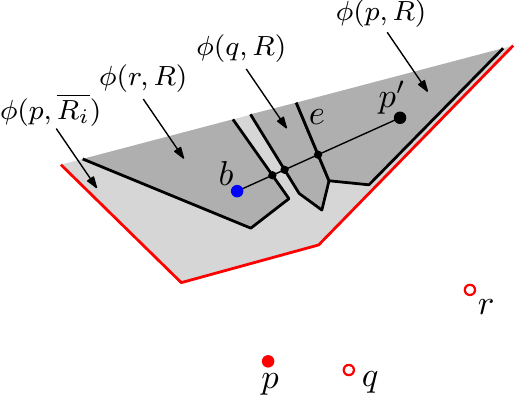}
\caption{$p$ is a point of $\ovl{R_i}$ that is farthest to $b$, and $r$ is a point of $R$ that is farthest to $b$.}
\label{fvc-fig}
\end{figure}
Consider the farthest point Voronoi diagram $\fvd{\ovl{R_i}}$. Because $p$ is a point of $\ovl{R_i}$ that is farthest to $b$, $b$ lies in $\fvc{p}{\ovl{R_i}}$, which is a convex region. See Figure~\ref{fvc-fig}. Imagine the construction of $\fvd{R}$ by inserting the points of $R_i$ to $\fvd{\ovl{R_i}}$. By inserting these points, the Voronoi cells of $\fvd{\ovl{R_i}}$ do not get larger. Thus, $\fvc{p}{R}$ is a subset of $\fvc{p}{\ovl{R_i}}$, and since $R$ is in convex position, $\fvc{p}{R}$ is not empty. Recall that $b$'s farthest red point in $R$, say $r$, is in $R_i$, and hence $b$ lies in $\fvc{r}{R}$. 
Take a point $p'$ in $\fvc{p}{R}$. 
Because of the convexity of $\fvc{p}{\ovl{R_i}}$, $bp'$ is inside this cell, and thus, no edge of $\fvd{\ovl{R_i}}$ crosses $bp'$. 
In $\fvd{R}$, $b$ and $p'$ belong to two different cells. This implies that some edges of $\fvd{R}$ cross $bp'$. Among those edges, consider the edge $e$ whose intersection with $bp'$ is closest to $p'$. Let $q\in R$ be the point such that $e$ is the common edge between $\fvc{q}{R}$ and $\fvc{p}{R}$. Since by inserting the points of $R_i$ into $\fvd{\ovl{R_i}}$, the cells of the points of $\ovl{R_i}$ get smaller, the point $q$\textemdash which has $e$ on its boundary in $\fvc{q}{R}$\textemdash belongs to $R_i$. This finishes the proof for the correctness of All-Blue-BFP. 

By a similar reasoning as in the proof of Lemma~\ref{size-T-lemma} we show that the total size of the sets $T_1,\dots, T_k$ is $O(|R|)$. 
The number of sets $T_i$, that a red point $p$ belongs to is at most $p$'s degree in $\fvd{R}$, that is, the number of its Voronoi neighbors. The sum of the degrees, over all points of $R$ is at most $4|R|-6$ because $\fvd{R}$ has at most $2|R|-3$ edges (or pairs of neighbors). Therefore, the total size of the sets $T_1,\dots,T_k$ is $O(|R|)$. Moreover, having $\fvd{R}$, these sets can be computed in $O(|R|)$ time by checking all edges of $\fvd{R}$.

The running time analysis of algorithm All-Blue-BFP is similar to the one for algorithm All-Blue-BCP that we have seen in Section~\ref{n-1-section}. Instead of maintaining Delaunay triangulations, here we maintain convex hulls and farthest point Voronoi diagrams. First, we review some known results. Let $P$ and $Q$ denote two convex polygons in the plane. Aggarwal~\etal \cite{AAggarwal1989} showed how to compute the farthest point Voronoi diagram of $P$ in $O(|P|)$ time. Edelsbrunner \cite{Edelsbrunner1985} showed that the two vertices that define the maximum distance between $P$ and $Q$ can be computed in $O(|P|+|Q|)$ time; he also proved the lower bound of $\Omega(|P|+|Q|)$ for this problem. 
Assume the vertices of $P\cup Q$ are in convex position. Then, having $\CH{P}$ and $\CH{Q}$, one can {\em merge} them to compute $\CH{P\cup Q}$ in $O(|P|+|Q|)$ time. The reverse operation can also be done in linear time. That is, one can {\em split} $\CH{P\cup Q}$ to obtain $\CH{P}$ and $\CH{Q}$ in time $O(|P|+|Q|)$. 

Now we analyze the running time of algorithm All-Blue-BFP. It takes $O(|R|\log |R|)$ time to construct $\fvd{R}$. In the first stage of Bor\r{u}vka's algorithm we compute $\CH{B_i}$ and $\CH{T_i}$ for all $i\in\{1,\dots,k\}$. This takes $O(n\log n)$ time because the total size of the sets $B_1,\dots,B_k$ is $|B|$, and the total size of the sets $T_1,\dots,T_k$ is $O(|R|)$. Having $\fvd{R}$, $\CH{B_i}$ and $\CH{T_i}$ for all $i\in\{1,\dots,k\}$, we maintain $\CH{B_i}$'s and $\CH{T_i}$'s for the next stage of Bor\r{u}vka's algorithm, using similar merge and split operations as in the MinBST algorithm. This maintenance takes $O(n)$ time per level of the imaginary tree $\mathcal{T}$ which has height $O(\log n)$. Therefore, in total, the algorithm All-Blue-BFP takes $O(n\log n)$ time for all stages of Bor\r{u}vka's algorithm. 

\begin{theorem}
 \label{MaxBST-thr}
Given two disjoint sets $R$ and $B$ of points in the plane, a Euclidean maximum spanning tree in $K(R,B)$ can be computed in $\Theta(n\log n)$ time, where $n=|R\cup B|$.
\end{theorem}
\section{The Min-$k$-ST and Max-$k$-ST problems}
\label{k-color-section}
In a multicolored version of the Euclidean minimum spanning tree problem, the input points are colored by at least two colors, and we want the colors of the two endpoints of every edge in the tree to be distinct. Formally, we are given a set $P$ of $n$ points in the plane that is partitioned into $\{P_1,\dots,P_k\}$, with $k \geqslant 2$. For each $c\in\{1,\dots,k\}$, assume the points of $P_c$ are colored $c$. In the bichromatic setting, $k$ is $2$. Also, the standard Euclidean minimum spanning tree problem can be interpreted as an instance of this multicolored version in which $k=n$, i.e., each point has a unique color.  
Let $K(P_1,\dots,P_k)$ be the complete multipartite geometric
graph on $P$, which has edges between every point of each set in the partition to all points of the other sets. The Min-$k$-ST problem is the problem of computing a minimum spanning tree in $K(P_1,\dots,P_k)$. We refer to its maximum counterpart as Max-$k$-ST. In this section we show how to solve these problems in $O(n\log n \log k)$ time. We show this for the Min-$k$-ST problem; the solution for the Max-$k$-ST problem is analogous.

The algorithm is as follows. Represent each of the colors $1,\dots, k$ in binary as a sequence of $\log k$ bits. For each $i\in\{1,\dots, \log k\}$, define two {\em canonical point sets}: a point set $\cset{i}$ that contains the points whose color's $i$th bit is $1$, and a point set $\ncset{i}$ that contains the points whose color's $i$th bit is $0$. Note that $\{\cset{i},\ncset{i}\}$ is a partition of $P$. 
In order to compute a Min-$k$-ST in $K(P_1,\dots,P_k)$, we do the following. First we compute MinBST($\cset{i},\ncset{i}$) for all $i\in\{1,\dots, \log k\}$. This gives $\log k$ spanning trees, each of which has $n-1$ edges (these edges are bichromatic in terms of $\cset{i}$ and $\ncset{i}$). Construct a graph $G$ by taking the union of these spanning trees; $G$ has $(n-1)\log k$ edges. Then, compute a minimum spanning tree of $G$ by Prim's algorithm, and output it as a solution. 

The running time analysis of this algorithm is straightforward. Since the construction of MinBST($\cset{i},\ncset{i}$) takes $O(n\log n)$ time for each $i$, the total time to compute all MinBSTs is $O(n\log n\allowbreak \log k)$. The running time of Prim's algorithm on $G$ is $O(n\log n+ n\log k)$. Thus, the total running time of the algorithm is $O(n\log n \log k)$. To prove the correctness of the algorithm, it suffices to show that $G$ contains an optimal Min-$k$-ST of $K(P_1,\dots,P_k)$. We show this in the following lemma. 

\begin{lemma}
 \label{k-color-lemma}
There exists an optimal Min-$k$-ST of $K(P_1,\dots,P_k)$ that is a subgraph of $G$. 
\end{lemma}
\begin{proof}
 The proof is constructive.  
\let\qed\relax\end{proof}
\vspace{-9pt}
\begin{wrapfigure}{r}{2.2in} 
\vspace{-8pt} 
\includegraphics[width=2.2in]{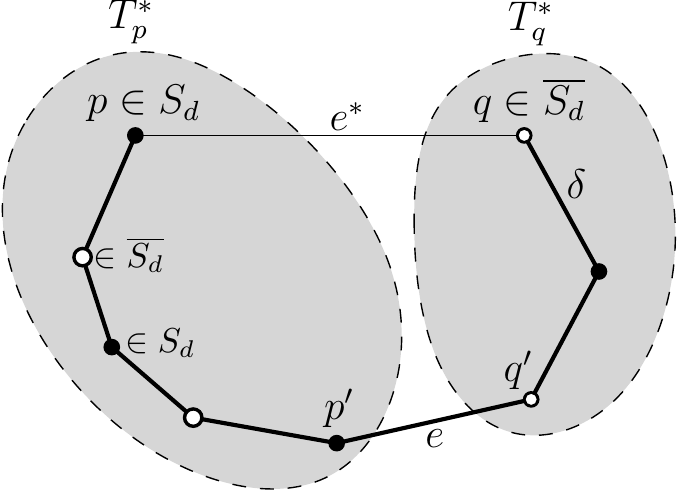} 
\vspace{-15pt} 
\end{wrapfigure}
Take an optimal Min-$k$-ST of $K(P_1,\dots,P_k)$, say $T^*$. Let $E^*$ be the set of edges of $T^*$ that are not in $G$. If $E^*$ is empty, then $T^*$ is a subgraph of $G$ and we are done. Otherwise, we iteratively replace the edges of $T^*$, that are in $E^*$, by some edges of $G$, without increasing the weight of $T^*$. Take any edge $e^*=(p,q)$ in $E^*$.
Assume $p$ is colored $i$ and $q$ is colored $j$. Since $i\neq j$, the binary representations of $i$ and $j$ differ in at least one bit. Let $d$ be such a bit position. Without loss of generality assume that $i$'s $d$th bit is $1$ and $j$'s $d$th bit is $0$. This implies that $p\in \cset{d}$ and $q\in\ncset{d}$. Let $T_d$ be the tree that is obtained by MinBST($\cset{d},\ncset{d}$); note that $T_d$ is a spanning subgraph of $G$. Since $e^*\notin G$, $e^*\notin T_d$. Consider the path $\delta$ between $p$ and $q$ in $T_d$. Since $T_d$ is an optimal bichromatic tree between $\cset{d}$ and $\ncset{d}$, none of the edges of $\delta$ is longer than $e^*$. Let $T^*_p$ and $T^*_q$ be the two trees obtained from $T^*$ by removing $e^*$. There is an edge $e$ in $\delta$ such that $e\notin T^*$ and $e$ connects a point $p'\in T^*_p$ to a point $q'\in T^*_q$. Let $T$ be the tree that is obtained by joining $T^*_p$ and $T^*_q$ with $e$. The tree $T$ is a valid spanning tree of $K(P_1,\dots,P_k)$ whose weight is equal or smaller than the weight of $T^*$. Therefore, $T$ is also an optimal Min-$k$-ST and the number of edges of $T$ that do not belong to $G$ is $|E^*|-1$. Set $T^*=T$, and repeat this process until $E^*=\emptyset$.   
\qed

\begin{theorem}
 \label{kBST-thr}
Given a set $P$ of $n$ points in the plane where $P$ is partitioned into $k\geqslant 2$ sets $P_1,\dots, P_k$, a Euclidean minimum $($respectively maximum$)$ spanning tree in $K(P_1,\dots,P_k)$ can be computed in $O(n\log n\log k)$ time.
\end{theorem}

\section{Open Problem}
We presented $O(n\log n\log k)$-time algorithms for the Min-$k$-ST and Max-$k$-ST problems. Presenting faster algorithms for these problems or providing a matching lower bound is open. Notice that for $k\in\{2,n\}$, these problems can be solved in $O(n\log n)$ time.

\bibliographystyle{abbrv}
\bibliography{Bichromatic-Trees.bib}
\end{document}